\documentclass[letterpaper, 10 pt, conference]{ieeeconf}
\IEEEoverridecommandlockouts
\let\labelindent\relax
\usepackage{enumitem}
\usepackage{romannum}
\usepackage{textcomp}
\usepackage{balance}
\usepackage[font={small}]{caption}
\usepackage{subcaption}
\usepackage{array}
\usepackage{textcomp}
\usepackage{mathtools, nccmath}
\usepackage{graphicx}
\usepackage{amsfonts}
\usepackage{amsmath}
\usepackage{amssymb}
\usepackage{algorithm}
\usepackage{algorithmic}
\usepackage{hyperref}
\usepackage{tikz}
\usepackage{arydshln}
\usepackage{multirow}
\usepackage{bm}
\usepackage{epstopdf}
\usepackage{cite}
\usepackage{siunitx}
\usepackage{lipsum}
\usepackage{wrapfig}

\usepackage{amsthm}

\newtheorem{theorem}{Theorem}

\newtheorem{remark}{Remark}
\theoremstyle{definition}
\newtheorem{definition}{Definition}
\newtheorem{problem}{Problem}

\title{\LARGE \bf
        A Predictive and Sampled-Data Barrier Method for Safe and Efficient Quadrotor Control  
}

\author{Ming Gao$^{1}$,  Zhanglin Shangguan$^{1}$, Shuo Liu$^{2}$, Liang Wu$^{3}$, Bo Yang$^{1}$, Wei Xiao$^{3}$
\thanks{$^{1}$Ming Gao, Zhanglin Shangguan and Bo Yang are with the School of Automation and Intelligent Sensing, Shanghai Jiao Tong University, Shanghai 200240,  China
        {\tt\small (e-mail: ming\_gao@sjtu.edu.cn; ditto331@sjtu.edu.cn; bo.yang@sjtu.edu.cn)}}%
\thanks{$^{2}$Shuo Liu is with the Dept. Mechanical Engineering, Boston University, Brookline, MA, 02215, USA
        {\tt\small (e-mail: liushuo@bu.edu)}}%
\thanks{$^{3}$Liang Wu and Wei Xiao are with MIT, Cambridge, MA 02139, USA 
        {\tt\small (e-mail: liangwu@mit.edu, weixy@mit.edu)}}%
}

\begin{document}

\maketitle
\thispagestyle{empty}
\pagestyle{empty}

\begin{abstract}
This paper proposes a cascaded control framework for quadrotor trajectory tracking with formal safety guarantees. First, we design a controller consisting of an outer-loop position model predictive control (MPC) 
and an inner-loop nonlinear attitude control, enabling decoupling of position safety and yaw orientation. Second, since quadrotor safety constraints often involve high relative degree, we adopt high order control barrier functions (HOCBFs) to guarantee safety. To employ HOCBFs in the MPC formulation that has formal guarantees, we extend HOCBFs to sampled-data HOCBF (SdHOCBFs) by introducing compensation terms, ensuring safety over the entire sampling interval. We show that embedding SdHOCBFs as control-affine constraints into the MPC formulation guarantees both safety and optimality while preserving convexity for real-time implementations. Finally, comprehensive simulations are conducted to demonstrate the safety guarantee and high efficiency of the proposed method compared to existing methods.
\end{abstract}

\section{INTRODUCTION}
Quadrotors are becoming increasingly prevalent in many applications~\cite{emranquadrotorsurvey}, including agricultural inspection, delivery, search-and-rescue, and many others. As these robots usually operate in extremely complex environments, ensuring both control performance and safety has become paramount. Despite extensive works on safety and optimal control of quadrotors, real-time implementation of controllers with formal safety guarantees remains a challenge either due to the nonlinearity and nonconvexity of optimization, or due to the compromise on optimality. 

As the coupling of position and attitude dynamics, quadrotors are commonly controlled using cascaded control architectures that exploit the differential flatness property~\cite{kumardifferentialflatness} to simplify controller design. The authors in~\cite{Dariotiltprior} further demonstrate that the position and yaw orientation can be independently controlled through a tilt-prioritized attitude controller. However, the decoupling of position and yaw orientation has not yet been explored in the context of safe trajectory tracking for quadrotors. Due to its ability to explicitly handle constraints and predict future system information, model predictive control (MPC) is widely used in quadrotors(see~\cite{romerompcc,hanoverL1mpcc,suncomparativesurvey}). In~\cite{Andriencascadedmpc}, the authors propose a cascaded control approach with MPC as outer-loop controller, achieving almost global trajectory tracking guarantees and real-time implementation on embedded hardware. However, this method does not provide a principled way to guarantee safety of quadrotors.

Safety in the MPC formulation is typically enforced by formulating state constraints as inequality conditions~\cite{fraschMPCDC}. However, such a MPC formulation is generally nonconvex and thus cannot ensure a solution that can satisfy all the constraints while incurring additional computational burden. Linearizations in a nonlinear MPC can improve efficiency and facilitate the derivation of solutions \cite{wangmpchocbf}, but it may make the method lose safety guarantees. Transforming state constraints into control-affine linear constraints, Control barrier functions (CBFs)~\cite{amescbf} have emerged as a powerful alternative for ensuring safety effectively. Traditional CBFs can only be used for constraints with relative degree one, which are constrast to the scenarios of quadrotors where most obstacles constraints are high relative degree. In~\cite{xiaohocbf}, high order control barrier functions (HOCBFs) are introduced to handle safety constraints with arbitrary relative degrees. However, HOCBFs are developed specifically for continuous-time systems and cannot be directly applied in MPC frameworks. To apply CBFs in sampled-data systems with zero-order-hold (ZOH) inputs that are commonly used in MPC formulation, sampled-data CBFs for constraints with relative degree one have been proposed in~\cite{Andrewzohcbf,Breedenzohcbf,agrawaldiscretecbf}. The authors in~\cite{zengmpccbf,alilinearmpccbf} then incorporate the discrete control barrier functions (DCBFs) from~\cite{agrawaldiscretecbf} as constraints within MPC formulation to ensure obstacle avoidance, leading to a general nonconvex and nonlinear optimization. In~\cite{wangmpchocbf}, HOCBFs are used as a filter for MPC to enforce safety while maintaining computational feasibility. Yet, neglecting the sampled-data systems commonly present in MPC, this method lacks theoretical safety guarantees and compromises the optimality of the MPC solution.

To address these challenges, this paper proposes a cascaded quadrotor trajectory tracking control framework that embeds HOCBFs as constraints into outer-loop MPC by extending HOCBFs into sampled-data HOCBFs (SdHOCBFs), enabling real-time, optimal, and formally safe control. The contributions of this paper are at least three folds:

1) We design a control framework using MPC as outer-loop controller with an inner-loop attitude controller to achieve quadrotor trajectory tracking and the decoupling of position safety and yaw orientation. 

2) We extend HOCBFs to SdHOCBFs to address its incompatibility with sampled-data system with ZOH inputs. SdHOCBFs can guarantee safety over the entire sampling interval and be embedded into MPC as control-affine constraints, ensuring both safety and optimality while preserving convexity for real-time implementation.

3) We verify our control framework on a quadrotor and compare it with sate-of-the-art methods to show safety, efficiency, and performance improvement of our method.

\section{PRELIMINARIES AND PROBLEM FORMULATION}
\label{sec:preliminary and problem formulation}
\subsection{Quadrotor Model}
Let $\bm{x} = [\bm{p}_{IB}^\top,\bm{v}_{IB}^\top,\bm{q}_{IB}^\top, \bm{\omega}_B^\top]^\top$ be the quadrotor's state space, where $I$ represents the inertial frame, $B$ represents the body frame and the subscript $IB$ represents the transformation from the body frame to the inertial frame, $\bm{p}_{IB} \in \mathbb{R}^3, \bm{q}_{IB} \in \mathbb{SO}(3), \bm{v}_{IB} \in \mathbb{R}^3,\bm{\omega}_{B} \in \mathbb{R}^3$ are the position, unit quaternion describing the attitude of $B$ relative to $I$, velocity and body rate in the body frame, respectively. The control input is $\bm{u}=[\bm{f}_B^\top,\bm{\tau}_B^\top]^\top$, where $\bm{f}_{B}=[0,0,f_{Bz}]^\top\in \mathbb{R}^3$ is the collective thrust  and $\bm{\tau}_B \in \mathbb{R}^3$ is the body torque. For simplicity, we drop the frame indices as they are consistent throughout the paper. The quadrotor's dynamics~\cite{romerompcc} are 
\begin{subequations}
    \label{eq: quadcopter dynamics}
    \begin{align}
        \dot{\bm{p}} &= \bm{v,} \label{eq: position-dot-eq} \\
        \dot{\bm{v}} &= \bm{g} + \frac{1}{m}\bm{R}(\bm{q})\bm{f} - \bm{D}\bm{v}, \label{eq: velocity-dot-eq}\\
        \dot{\bm{q}} &= \frac{1}{2}\bm{q} \odot \begin{bmatrix}
            0 \\ \bm{\omega}
        \end{bmatrix},  \label{eq: quaternion-dot-eq} \\
        \dot{\bm{\omega}} &= \bm{J}^{-1}(\bm{\tau}-\bm{\omega} \times \bm{J}\bm{\omega}) \label{eq: omega-dot-eq},
    \end{align}
\end{subequations}
where $\odot$ represents the Hamilton quaternion multiplication, $\bm{R}(\bm{q}) \in \mathbb{SO}(3)$ is the rotation matrix from unit quaternion $\bm{q}$, $m$ is the mass, $\bm{g}$ is the gravitational acceleration, $\bm{J} \in \mathbb{R}^{3\times 3}$ is the quadrotor's inertia matrix and $\bm{D}=$diag$(d_x,d_y,d_z)$ is the aerodynamic drag matrix. 

\subsection{Safety via High Order Control Barrier Functions}
In this section, we introduce the relevant background on HOCBFs. Consider an affine control system expressed as 
\begin{equation}
    \dot{\bm{x}} = f(\bm{x})+g(\bm{x})\bm{u},
    \label{eq:affine sys}
\end{equation}
where $\bm{x} \in \mathcal{X} \subset \mathbb{R}^n$, $\bm{u} \in \mathcal{U} \subset \mathbb{R}^m$, $f:\mathbb{R}^n \rightarrow \mathbb{R}^n$ and $g:\mathbb{R}^n \rightarrow \mathbb{R}^{n \times m}$ are locally Lipschitz on $\mathcal{X}$. Moreover, $\mathcal{X}$ and $\mathcal{U}$ denote the admissable states and inputs, respectively.

The safe set $\mathcal{C} \subseteq \mathcal{X}$ is defined as the zero-superlevel set of a sufficiently differentiable function $h: \mathbb{R}^n \rightarrow \mathbb{R}$, i.e. $\mathcal{C} \coloneqq \lbrace \bm{x} \in \mathbb{R}^n : h(\bm{x}) \geq 0 \rbrace$.

For $h:\mathbb{R}^n \rightarrow \mathbb{R}$ of relative degree $\rho$, i.e. $h$ must be differentiated along~\eqref{eq:affine sys} at least $\rho$ times to make control input $\bm{u}$ appears explicitly,  we define a series of functions:
\begin{equation}
    \label{eq: hocbf functions series}
    h_i(\bm{x}) = L_f{h}_{i-1}(\bm{x}) + \alpha_i(h_{i-1}(\bm{x})), i \in \lbrace 1,\cdots,\rho-1 \rbrace, 
\end{equation}
with $h_0(\bm{x})=h(\bm{x})$, where $\alpha_i(\cdot),i \in \lbrace 1,\cdots,\rho-1 \rbrace$ are class $\mathcal{K}$ functions, $L_f$ and $L_g$ denote the Lie derivatives along $f$ and $g$, respectively. Note that $L_g$ does not appear in the recursive definitions \eqref{eq: hocbf functions series} because, by construction, the control input enters explicitly only at the final step \eqref{eq: HOCBF}. We also define a series of sets
\begin{equation}
    \label{eq: hocbf sets series}
    \mathcal{C}_i = \lbrace \bm{x} \in \mathbb{R}^n : h_i(\bm{x}) \geq 0 \rbrace, i=\lbrace 0,\cdots,\rho-1 \rbrace. 
\end{equation}

\begin{definition}[HOCBF~\cite{xiaohocbf}]
    \label{def: hocbfdef}
    Given $h_i, i\in\{0, \dots, \rho-1\}$ as in (\ref{eq: hocbf functions series}) with the corresponding sets $\mathcal{C}_i, i\in\{0, \dots, \rho-1\}$ defined as in (\ref{eq: hocbf sets series}), $h:\mathbb{R}^n \rightarrow \mathbb{R}$ is a HOCBF of relative degree $\rho$ for system \eqref{eq:affine sys} if there exist differentiable class $\mathcal{K}$ functions $\alpha_{i}, i\in\{1,\dots, \rho\}$ such that
    \begin{equation}
            \sup_{\bm{u} \in \mathcal{U}} [ L_fh_{\rho-1}(\bm{x}) + L_gh_{\rho-1}(\bm{x})\bm{u} + \alpha_{\rho} (h_{\rho-1}(\bm{x}))] \geq 0,
        \label{eq: HOCBF}
    \end{equation}
    for all $\bm{x} \in \mathcal{C}_0 \cap,\cdots,\cap \mathcal{C}_{\rho-1}$.
\end{definition}
\begin{theorem}
    \label{thm: hocbf}
    \cite{xiaohocbf} Given an HOCBF $h(\bm{x})$ from Def.~\ref{def: hocbfdef} with the sets $\mathcal{C}_i, i \in \lbrace 0, \cdots, \rho-1 \rbrace$, if $\bm{x}(t_0) \in \mathcal{C}_0 \cap \cdots \cap \mathcal{C}_{\rho-1}$, then any Lipschitz continuous controller $\bm{u}(t)$ satisfying \eqref{eq: HOCBF} $\forall t > t_0$ renders the set $\mathcal{C}_0 \cap \cdots \mathcal{C}_{\rho-1}$ forward invariant for system \eqref{eq:affine sys}, i.e. $\bm{x} \in \mathcal{C}_0 \cap \cdots \mathcal{C}_{\rho-1}, \forall t > t_0$.  
\end{theorem}

\subsection{Problem Formulation}
This paper considers the problem of trajectory tracking for a quadrotor in environments with potential obstacles. We mainly address the following problem:
\begin{problem}
    \label{prob: problem formulation}
    Given the quadrotor dynamics \eqref{eq: quadcopter dynamics} and a feasible reference trajectory denoted as $\bar{\bm{y}} = [\bar{\bm{p}}^T, \bar{\psi}]^T$, where $\bar{\psi}$ is the reference yaw angle, find a control law $\bm{u}(\bm{x},\bar{\bm{y}})$ such that the resulting inputs $\bm{u}$ can minimize the following cost:
    \begin{equation}
        \label{eq : stable problem}
    \min_{\bm{u}(t)} \int_{t_0}^{t_f}  \mathcal{Q}(\bm{x}(\tau), \bar{\bm{y}}(\tau)) + \mathcal{R}(\bm{u}(\tau), \bar{\bm{y}}(\tau)) d\tau,
    \end{equation}
    where $\mathcal{Q}:\mathbb{R}^{13}\times\mathbb{R}^{4}\rightarrow \mathbb{R}$ and $\mathcal{R}:\mathbb{R}^{6}\times\mathbb{R}^{4}\rightarrow \mathbb{R}$ are both positive definite functions and $t_f>t_0$.
    
    In addition, the quadrotor should always satisfy a safety requirement and control bounds 
    \begin{equation}
        \label{eq: safe problem}
        h(\bm{x}(t)) \geq 0, \forall{t} \in [t_0,t_f],
    \end{equation}
    \begin{equation}
        \label{eq: control bounds problem}
        0 < f_{z} \leq f_{max}, \forall{t} \in [t_0,t_f],
    \end{equation}
    where $h:\mathbb{R}^{13}\rightarrow\mathbb{R}$ is continuously differentiable and has relative degree $\rho \in \mathbb{N}$ with respect to the quadrotor~\eqref{eq: quadcopter dynamics}, and $f_z = f_{Bz}$ as we drop the frame indices for simplicity.
\end{problem}

Several challenges must be addressed to effectively tackle the abvoe problem: a) ensuring safety without incurring excessive computational cost; b) achieving both optimality and safety. The remainder of this paper is devoted to address Problem~\ref{prob: problem formulation} with these challenges.

\section{METHODOLOGY}
\label{sec:methodology}

To address Problem~\ref{prob: problem formulation}, a cascaded control framework is proposed, consisting of an outer-loop position MPC controller and an inner-loop nonlinear attitude controller, as depicted in Fig~\ref{fig: controller structure}. The outer-loop MPC controller imposes constraints on the acceleration to account for thrust limitations of the quadrotor, while incorporating high order control barrier functions (HOCBFs) at the first prediction step to ensure safety constraints are satisfied. The output of the outer-loop controller is then transformed into the desired attitude angles, as required by the inner-loop controller, through differential flatness~\cite{kumardifferentialflatness, Faesslerdifferentialflatness}.
\subsection{Outer-loop MPC Controller}
\subsubsection{Model}
The outer-loop MPC controller provides desired accelerations for the inner-loop controller, only considering position dynamics \eqref{eq: position-dot-eq}~\eqref{eq: velocity-dot-eq}. However, given the differential flatness property, the desired accelerations must be twice differentiable to get desired attitudes and angular velocities, which will be shown in Section~\ref{subsec:conversion}. Therefore, we augment the position dynamics \eqref{eq: position-dot-eq}~\eqref{eq: velocity-dot-eq}. Consider the following augmented system:
\begin{subequations}
    \label{eq: augment position dynamics}
    \begin{align}
        \dot{\bm{p}} &= \bm{v},  \label{eq: augment postion dynamics dot p} \\
        \dot{\bm{v}} &= -\bm{D}\bm{v} +\bm{a}_v, \label{eq: augment postion dynamics dot v} \\
        \dot{\bm{a}}_v &= \bm{j}_v, \label{eq: augment postion dynamics dot ad} \\
        \dot{\bm{j}}_v &= \bm{s}_v, \label{eq: augment postion dynamics dot j}
    \end{align}
\end{subequations}
where $\bm{p},\bm{v},\bm{a}_v,\bm{j}_v\in \mathbb{R}^3$, $\bm{a}_v$ is a virtual state to replace $\bm{g} + \frac{1}{m}\bm{R}(\bm{q})\bm{f}$, and $\bm{s}_v \in \mathbb{R}^3$ is considered as the input.

To employ system \eqref{eq: augment position dynamics} in MPC, we discretize it with a fixed time step into the following form:
\begin{equation}
    \label{eq: discreted augment position dynamics}
    \bm{z}_{k+1} = A\bm{z}_{k} + B\bm{s}_{v,k}, k \in \mathbb{N}_0,
\end{equation}
where $\bm{z}_{k} = \begin{bmatrix} \bm{p}_k^\top & \bm{v}_k^\top & \bm{a}_{v,k}^\top & \bm{j}_{v,k}^\top \end{bmatrix}^\top$ and $\bm{s}_{v,k}$ are the state and the input at time step $k$, respectively. $A \in \mathbb{R}^{12 \times 12}$ and $B \in \mathbb{R}^{12 \times 3}$ can be obtained analytically.

In addition, the virtual state $\bm{a}_v$ must satisfy some constraints to satisfy the thrust limitations~\eqref{eq: control bounds problem}. Given the fact that $\lVert\bm{R}(\bm{q})\bm{f}\lVert=f_z$,
we can determine the constraints on $\bm{a}_v$, i.e.
\begin{equation}
    \label{eq: thrust cone constraints}
    m\lVert \bm{a}_v-\bm{g} \rVert \leq f_{max}.
\end{equation}

\begin{figure}[t]
    \centering
    \includegraphics[scale=0.35]{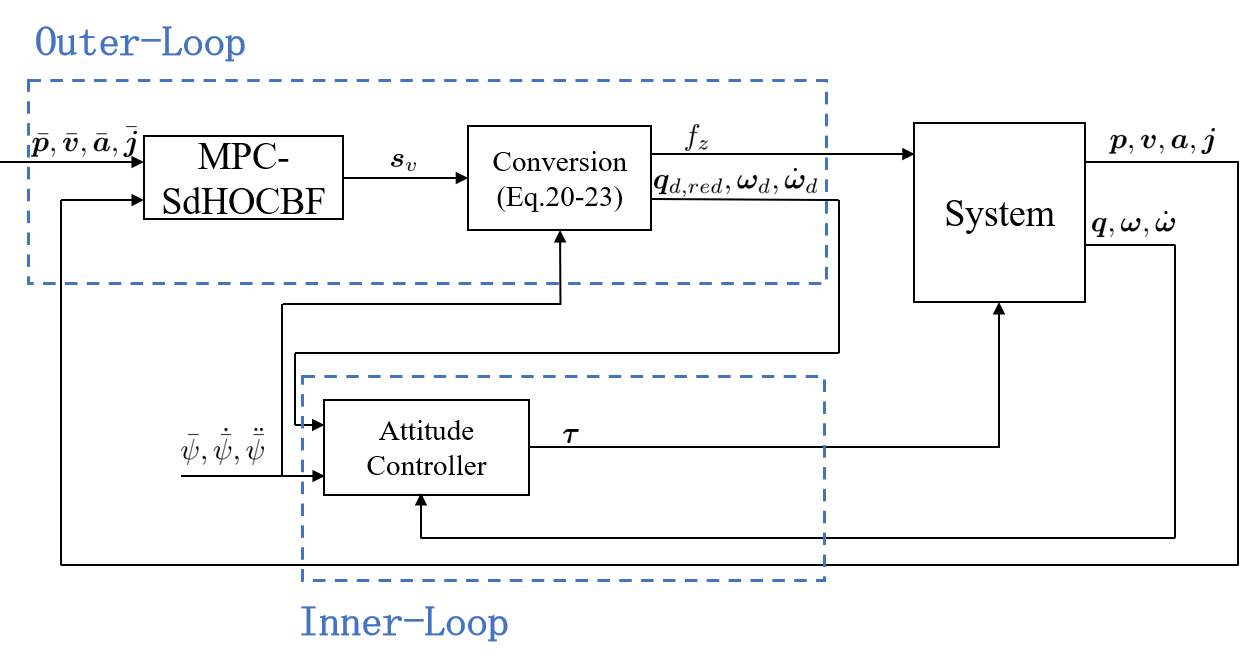}
    \caption{Diagram of the proposed cascaded controller for safe quadrotor control. MPC-SdHOCBF takes position states to generate $\bm{s}_v$, which is converted into thrust $f_z$ and angular states. The attitude controller then uses the angular states to compute torques $\bm{\tau}$.}     
    \label{fig: controller structure}
\end{figure}

\subsubsection{HOCBFs for Sampled Data System}
In this paper, we employ HOCBFs \eqref{eq: HOCBF} to ensure safety of quadrotor based on the augmented system~\eqref{eq: augment position dynamics}. However, HOCBFs \eqref{eq: HOCBF} are formulated for continuous-time systems, only guaranting the safety at current time, while the system \eqref{eq: discreted augment position dynamics} used in MPC framework is a sampled-data system with zero-order-hold (ZOH) inputs. Therefore, it's necessary to modify the HOCBFs to guarantee the safety of such systems.

Given an affine control system \eqref{eq:affine sys}, to apply Thm.~\ref{thm: hocbf}, we must ensure that \eqref{eq: HOCBF} is satisfied for all $t \geq t_0$. However, in sampled-data system with ZOH inputs, we can only get the state $\bm{x}$ at $t=kT,k \in \mathbb{N}_0$ for a fixed time step $T$ and the inputs $\bm{u}(t)$ remain constant for $t \in [kT,(k+1)T), k \in \mathbb{N}_0$. Thus, given a safe set $\mathcal{C}=\lbrace \bm{x} \in \mathbb{R}^n:h(\bm{x}) \geq 0 \rbrace$,  we need to find $\bm{u}_k$ to make the following constraint satisfied:
\begin{equation}
    \label{eq: sampled h constraints}
    h(\bm{x}) \geq 0, \forall t \in[kT,(k+1)T).
\end{equation}

To achieve this, we first consider the case of $h(\bm{x})$ of relative degree one with respect to \eqref{eq:affine sys}. We need to find a compensation term $\phi(\bm{x},T)$, such that any $\bm{u}_k$ satisfying
\begin{equation}
    \label{eq: sampled-data CBF}
    L_fh(\bm{x}_k) + L_gh(\bm{x}_k)\bm{u}_k \geq -\alpha(h(\bm{x}_k)) + \phi(\bm{x}_k,T),
\end{equation}
can guarantee \eqref{eq: sampled h constraints}.

Inspired by~\cite{Andrewzohcbf} and~\cite{Breedenzohcbf}, we define the compensation term $\phi(\bm{x},T)$ as 
\begin{equation}
    \label{eq: hocbf difference}
    \begin{aligned}
    \phi(\bm{x},T) = &\inf_{\bm{z} \in \mathcal{R}(T,\bm{x}), \bm{u} \in \mathcal{U}} L_fh(\bm{z}) + L_gh(\bm{z})\bm{u} + \alpha(h(\bm{z})) \\ 
                &- L_fh(\bm{x}) - L_gh(\bm{x})\bm{u} - \alpha(h(\bm{x})),
    \end{aligned}
\end{equation}
where $\mathcal{R}(T,\bm{x})$ denotes the set of states reachable from $\bm{x}$ in time interval $T$. Take \eqref{eq: hocbf difference} into \eqref{eq: sampled-data CBF}, we can get
\begin{equation}
    \label{eq: one order sampled data cbf}
    L_fh(\bm{x}) + L_gh(\bm{x})\bm{u}_k \geq -\alpha(h(\bm{x})),\forall t \in [kT,(k+1)T), 
\end{equation}  
which can guarantee \eqref{eq: sampled h constraints} and the resulting $\bm{u}_k,k\in \mathbb{N}_0$ renders $\mathcal{C}$ forward invariant along the trajectories of \eqref{eq:affine sys}.   

We now extend the above formulation to $h(\bm{x})$ of any relative degree with respect to \eqref{eq:affine sys}. For a HOCBF $h(\bm{x})$ of relative degree $\rho$, we define a series of functions and sets same as in \eqref{eq: hocbf functions series} and \eqref{eq: hocbf sets series}. Then, we have the following results:
\begin{definition}[SdHOCBF]
    \label{def: zohhocbf}
    Let $h_i:\mathbb{R}^n \rightarrow \mathbb{R}, i\in \lbrace 1,\cdots,\rho-1\rbrace$ be defined by \eqref{eq: hocbf functions series} and $\mathcal{C}_i, i\in \lbrace0,\cdots,\rho-1\rbrace$ be defined by \eqref{eq: hocbf sets series}. $h: \mathbb{R}^n \rightarrow \mathbb{R}$ is a sampled-data high order control barrier of relative degree $\rho$ for system~\eqref{eq:affine sys} with step-time $T$ if there exists class $\mathcal{K}$ functions $\alpha_{\rho}:\mathbb{R}\rightarrow \mathbb{R}$ such that 
    \begin{equation}
        \label{eq: zohhocbf definition}
        \begin{aligned}
        \sup_{\bm{u}_k \in \mathcal{U}} &[L_fh_{\rho-1}(\bm{x}_k) + L_gh_{\rho-1}(\bm{x}_k)\bm{u}_k + \phi(\bm{x}_k,T)\\
            &+ \alpha_{\rho} (h_{\rho-1}(\bm{x}_k))] \geq 0, 
        \end{aligned}
    \end{equation}
    \begin{equation}
        \label{eq: zohhocbf error}
        \phi(\bm{x}_k,T) = \inf_{\bm{z} \in \mathcal{R}(T,\bm{x}_k), \bm{u} \in \mathcal{U}} H(\bm{z},\bm{u})-H(\bm{x}_k,\bm{u}),
    \end{equation}
    \begin{equation}
        \label{eq: hocbf metric}
        H(\bm{x},\bm{u}) = L_fh_{\rho-1}(\bm{x}) + L_gh_{\rho-1}(\bm{x})\bm{u} + \alpha_{\rho} (h_{\rho-1}(\bm{x})),
    \end{equation}
    for all $\bm{x}_k \in \mathcal{C}_0 \cap \mathcal{C}_1 \cap \cdots \cap \mathcal{C}_{\rho-1}$.
\end{definition}
\begin{theorem}
    \label{thm: zohhocbf theorem}
    Let $h: \mathbb{R}^n \rightarrow \mathbb{R}$ defined by Def.~\ref{def: zohhocbf} with set $\mathcal{C} = \mathcal{C}_0 \cap \cdots \cap \mathcal{C}_{\rho-1}$. If $\bm{x}_{0} \in \mathcal{C}$, then any $\bm{u}_k$ satisfying \eqref{eq: zohhocbf definition} renders $\mathcal{C}$ forward invariant for system \eqref{eq:affine sys}. 
\end{theorem}
\begin{proof}
    If $h(\bm{x})$ is a SdHOCBF for system \eqref{eq:affine sys} defined by Def.~\ref{def: zohhocbf}, then for all $t \in [kT,(k+1)T),k \in \mathbb{N}_0$, it holds that
    \[
        \begin{aligned}
            &\quad L_fh_{\rho-1}(\bm{x}) + L_gh_{\rho-1}(\bm{x})\bm{u}_k \\ 
            &= H(\bm{x},\bm{u}_k) - \alpha_{\rho}(h_{\rho-1}(\bm{x}))\\
            &= H(\bm{x},\bm{u}_k) - H(\bm{x},\bm{u}_k) + H(\bm{x},\bm{u}_k) - \alpha_{\rho}(h_{\rho-1}(\bm{x})) \\
            &\geq H(\bm{x},\bm{u}_k) + \phi(\bm{x}_k,T) -\alpha_{\rho}(h_{z,\rho-1}(\bm{x}))\\
            &\geq -\alpha_{\rho}(h_{z,\rho-1}(\bm{x})).
        \end{aligned}
    \]
    By Thm. \ref{thm: hocbf}, since $\bm{x}_k \in \mathcal{C} \subseteq \mathcal{C}_{\rho-1}$, we have $\bm{x}(t) \in \mathcal{C}_{\rho-1}, \forall t \in [t_0, t_0+T)$. Iteratively, we can get $\bm{x}(t) \in \mathcal{C}_{i-1}, \forall t \in [t_0, t_0+T), i\in \lbrace0,\cdots,\rho-1\rbrace$. Therefore, the set $\mathcal{C}$ is forward invariant for system \eqref{eq:affine sys} for $t \in [t_0,t_0+T)$.
\end{proof}
\begin{remark}
  The reachable set $\mathcal{R}(T, \bm{x}_k)$ is generally difficult to compute exactly and online, we can only obtain a more conservative superset of $\mathcal{R}(T, \bm{x}_k)$ to replace it. However, since the system \eqref{eq: discreted augment position dynamics} we used in MPC is linear, whose exact reachable set  $\mathcal{R}(T, \bm{x}_k)$ can be computed online, we can therefore also compute $\phi(\bm{x}_k,T)$ analytically.
\end{remark}

\subsubsection{MPC Formulation} 
Based on above construction, we define following MPC Formulation:
\begin{equation}
    \label{eq: mpc controller}
    \begin{aligned}
        \min_{S_{v,k}} \quad &l_N(\bm{z}_{N|k},\bar{\bm{z}}_{N|k}) + \sum_{i=0}^{N-1} l(\bm{z}_{i|k},\bm{s}_{v,i|k},\bar{\bm{z}}_{i|k}, \bar{\bm{s}}_{v,i|k}), \\
        s.t.\quad &\bm{z}_{i+1|k} = A\bm{z}_{i|k} + B\bm{s}_{v,i|k}, i \in \lbrace 0,\cdots, N-1 \rbrace, \\
        & \bm{z}_{0|k} = \bm{z}_{k}, \\
        & \bm{s}_{v,min} \leq \bm{s}_{v,i|k} \leq  \bm{s}_{v,max}, i \in \lbrace 0,\cdots, N-1 \rbrace, \\  
        & m\Vert \bm{a}_{v,i|k} - \bm{g} \Vert \leq f_{max}, i \in \lbrace 0,\cdots, N-1 \rbrace, \\
        & L_fh_{\rho-1}(\bm{z}_{k}) + L_gh_{\rho-1}(\bm{z}_{k})\bm{s}_{v,0|k} + \phi(\bm{z}_{k},T) \\ 
        &+ \alpha_{\rho} (h_{\rho-1}(\bm{z}_{k})) \geq 0, 
    \end{aligned}
\end{equation}
where $S_{v,k} = [\bm{s}_{v,0|k}^\top, \cdots, \bm{s}_{v,N-1|k}^\top]^\top$ denotes the control sequence from MPC.
 $h(\cdot)$ is a SdHOCBF of relative degree $\rho$ with respect to~\eqref{eq: augment position dynamics}, Let $S_{v,k}^\star=[\bm{s}_{v,0|k}^{\star \top}, \cdots, \bm{s}_{v,N-1|k}^{\star \top}]^\top$ be the optimal solution of \eqref{eq: mpc controller}. At every execution step, only $\bm{s}_v = \bm{s}_{v,0|k}^\star$ is applied to~\eqref{eq: augment position dynamics}
. $l_N(\cdot),l(\cdot)$ are positive definite functions and we use $l_N(\bm{z},\bar{\bm{z}}) = (\bm{z}-\bar{\bm{z}})^\top P(\bm{z}-\bar{\bm{z}})$ and $l(\bm{z},\bm{s}_v,\bar{\bm{z}},\bar{\bm{s}}_v)=(\bm{z}-\bar{\bm{z}})^\top Q(\bm{z}-\bar{\bm{z}})+(\bm{s}_v-\bar{\bm{s}}_v)^\top R(\bm{s}_v-\bar{\bm{s}}_v)$ in this paper, where $P,Q \in \mathbb{R}^{12 \times 12},R\in \mathbb{R}^{3\times 3}$ are positive definite matrix and $\bar{\bm{z}}$ and $\bar{\bm{s}}_v$ are the reference states and inputs, respectively. 

\begin{theorem}
    \label{thm: mpc safety}
    The MPC controller \eqref{eq: mpc controller} is a second-order cone program (SOCP) and renders set $\mathcal{C}=\lbrace \bm{z} \in \mathbb{R}^{12}:h(\bm{z}) \geq 0 \rbrace$ forward invariant for system \eqref{eq: augment position dynamics}.
\end{theorem}   
\begin{proof}
    All the constraints and costs in \eqref{eq: mpc controller} are linear or second-order cone, thus the optimization formulation is a SOCP. At every execution step, only $\bm{s}_v = \bm{s}_{v,0|k}^\star$ is applied to~\eqref{eq: augment position dynamics}, which can ensure $h(\bm{z}) \geq 0, \forall t \in [kT,(k+1)T)$ when $h(\bm{z}_k)\geq 0$, given Thm.~\ref{thm: zohhocbf theorem}. Iteratively, the close-loop trajectories along the system \eqref{eq: augment position dynamics} always satisfy $h(\bm{z}) \geq 0$ if $h(\bm{z}_0) \geq 0$. Then $C$ is forward invarint for \eqref{eq: augment position dynamics}.
\end{proof}
\begin{remark}
While the input constraints in \eqref{eq: mpc controller} can be relaxed to $m(\bm{a}_{v,i|k}-\bm{g}) \leq \frac{f_{max}}{\sqrt{3}}\bm{1}, \forall i \in \lbrace 0,\cdots,N-1\rbrace$, thereby yielding a quadratic program (QP), such a formulation is conservative. Both QP and SOCP are convex optimization problems that admit polynomial-time solutions via interior-point methods. Although QP solvers typically exhibit superior computational efficiency, in our setting, the solution time of SOCP and QP are practically indistinguishable, a point that will be corroborated by experiments in Sec.~\ref{sec: simulations and discussions}. 
\end{remark}
\subsection{Inner-loop Controller}
\label{subsec:conversion}
Given $\bm{s}_{v,0|k}^\star$, we can get the desired position state by 
\begin{equation}
    \label{eq: desired augment states}
    \bm{z}_{d} = A\bm{z}_{k} + B\bm{s}_{v,0|k}^\star.
\end{equation}

We can then get desired thrust $f_{z}$ for system~\eqref{eq: quadcopter dynamics}
and desired attitude $\bm q_d$ by $\bm{z}_d =[\bm{p}_d^\top , \bm{v}_d^\top , \bm{a}_{v,d}^\top , \bm{j}_{v,d}^\top]^\top$ from~\eqref{eq: desired augment states}
and $\bar{\psi}$ from reference trajectory
\begin{subequations}
    \label{eq: desired attitude}
    \begin{align}
        f_{z} &= m\Vert \bm{a}_{v,d} - \bm{g} \Vert, \label{eq: desired thrust} \\
        \bm{q}_d &= \bm{q}_{d,yaw} \odot \bm{q}_{d,red}, \label{eq: desired attitude2} \\
        \bm{q}_{d,red} &= \begin{bmatrix} \cos(\frac{\varphi_{d,red}}{2}) & \bm{n}_{d,red}^\top\sin(\frac{\varphi_{d,red}}{2})\end{bmatrix}^\top, \label{eq: desired reduced attitude} \\
        \varphi_{d,red} &= -\arctan(\frac{\sqrt{a_{v,d,x}^2+a_{v,d,y}^2}}{a_{v,d,z}+g}), \label{eq: desired reduced attitude phi} \\
        \bm{n}_{d,red} &= \frac{1}{\sqrt{a_{v,d,x}^2+a_{v,d,y}^2}}\begin{bmatrix}-a_{v,d,y} & a_{v,d,x} & 0 \end{bmatrix}^\top, \label{eq: desired reduced attitude axis} \\
        \bm{q}_{d,yaw} &= \begin{bmatrix} \cos(\frac{-\bar{\psi}}{2}) & 0 & 0 & \sin(\frac{-\bar{\psi}}{2}) \end{bmatrix}^\top. \label{eq: desired yaw attitude}
    \end{align}
\end{subequations}
It shows that $\bm{q}_{d,red}$ is fully determined by $\bm{a}_{v,d}$, while $\bm{q}_{d,yaw}$ only depends on $\bar{\psi}$. Therefore, in our proposed MPC framework \eqref{eq: mpc controller}, the yaw orientation $\psi$, which is irrelevant to obstacle avoidance, is excluded in our consideration.

To get the desired angular velocities and accelerations, we first convert the virtual states into the real ones by 
\begin{subequations}
    \label{eq: desired states}
    \begin{align}
        \bm{a}_d &= -\bm{D}\bm{v}_d + \bm{a}_{v,d}, \label{eq: desired acc} \\
        \bm{j}_d &= -\bm{D}\bm{a}_d + \bm{j}_{v,d}, \label{eq: desired jerk} \\
        \bm{s}_d &= -\bm{D}\bm{j}_d + \bm{s}_{v,0|k}^\star, \label{eq: desired snap}
    \end{align}
\end{subequations}
where $\bm{a}_d$, $\bm{j}_d$, $\bm{s}_d$ are desired accelerations, jerks and snaps, respectively. Then, combined with the $\bar{\psi}$, $\dot{\bar{\psi}}$, $\ddot{\bar{\psi}}$ from reference trajectory, we can get the desired angular velocities $\bm{\omega}_d$ and angular accelerations $\dot{\bm{\omega}}_d$ by
\begin{subequations}
    \label{eq: desired omega and domega}
    \begin{align}
        \bm{\omega}_d &=  (\bm{R}(\bm{q}_d)^\top\dot{\bm{R}}(\bm{q}_d))^\vee, \label{eq: get desired omega} \\
        \dot{\bm{\omega}}_d &=  (\bm{R}(\bm{q}_d)^\top\ddot{\bm{R}}(\bm{q}_d)-(\bm{\omega}_d^\wedge)^2)^\vee, \label{eq: get desired domega}
    \end{align}
\end{subequations}
where $(\cdot)^\wedge$ and $(\cdot)^\vee$ are the standard isomorphism between vector and skew-symmetric matrix. Thus, we have all elements for the inner-loop controller. 

We then employ the attitude controller proposed in~\cite{Dariotiltprior}, where the quadrotor's attitude error is split into a reduced attitude error and a yaw error to prioritize the quadrotor's ability to achieve desired acceleration rather than orientation.

Given a desired attitude $\bm{q}_d$ and quadrotor's current attitude $\bm{q}$, we define the attitude error $\bm{q}_e$, reduced attitude error $\bm{q}_{e,red}$ and yaw error $\bm{q}_{e,yaw}$ as follows:
\begin{subequations}
    \label{eq: quaternion error definition}
    \begin{align}
        \bm{q}_e &= \bm{q}_d \odot \bm{q} = \begin{bmatrix}q_{e,w} &q_{e,x} &q_{e,y} &q_{e,z}\end{bmatrix}^\top, \label{eq: quaternion error} \\ 
        \bm{q}_{e,yaw} &= \frac{1}{\sqrt{q_{e,w}^2+q_{e,z}^2}} \begin{bmatrix} 
            q_{e,w} & 0 & 0 & q_{e,z}
        \end{bmatrix}^\top, \label{eq: quaternion yaw error} \\
        \bm{q}_{e,red} &= \bm{q}_{e,yaw}^{-1} \odot \bm{q}_{e}. \label{eq: quaternion red errow}
    \end{align}
\end{subequations}
Subsequently, the desired torques $\bm{\tau}$ can be obtained by the following control law:
\begin{equation}
    \label{eq: inner-loop controller}
    \bm{\tau} = k_{p,xy}\tilde{\bm{q}}_{e,red} + k_{p,z}sgn(q_{e,w}) \tilde{\bm{q}}_{e,yaw} +\bm{K}_d\bm{\omega}_e + \bm{\tau}_{ff},
\end{equation}
where $\bm{\omega}_e=\bm{R}(\bm{q}_e)\bm{\omega}_d - \bm{\omega}$ is the angular velocity error, $\bm{\tau}_{ff}=\bm{J}\dot{\bm{\omega}}_e - \bm{J}\bm{\omega} \times \bm{\omega}$ is the feed-forward term and 
$\tilde{\bm{q}}_{e,red}$ and $\tilde{\bm{q}}_{e,yaw}$ are the imaginary part of $\bm{q}_{e,red}$ and $\bm{q}_{e,yaw}$, i.e.  $\tilde{\bm{q}}=[q_x , q_y , q_z]^\top$. $\bm{K}_d=diag(k_{d,x},k_{d,y},k_{d,z})$ and $k_{p,xy}, k_{p,z}, k_{d,x},k_{d,y},k_{d,z}$ are all positive constants. 

By setting a relative high $k_{p,xy}$ than $k_{p,yaw}$, this controller can prioritize achieving safe trajectory tracking ability enabled by the outer-loop MPC controller \eqref{eq: mpc controller} and achieve almost globally asymptotically stable about $(\bm{q}_{e,red},\bm{\omega}_{e,red}) = (\bm{q}_I,0)$ and locally asymptotically stable about $(\bm{q}_{e,yaw},\bm{\omega}_{e,yaw}) = (\bm{q}_I,0)$. 

\begin{table}[t]
\caption{Quadrotor Parameters}
\label{tab: quadrotor parameters}
\begin{center}
\begin{tabular}{|c|c|c|}
\hline
Symbol & Definition & Value\\
\hline
$m$ & Mass of quadrotor & $0.468kg$\\
$J_{xx},J_{yy}$ & Inertia of $x,y$-axis & $4.856 \times 10^{-3}kg\cdot m^2$ \\
$J_{zz}$ & Inertia of $z$-axis & $ 8.801 \times 10^{-3}kg\cdot m^2$ \\
$T_{max}$ & Maximum total thrust & $12N$ \\
$\bm{D}$ & Aerodynamic drag matrix & $diag(0.25,0.25,0.25) s^{-1}$ \\
\hline
\end{tabular}
\end{center}
\end{table}

\section{SIMULATIONS AND DISCUSSIONS}
\label{sec: simulations and discussions}
In this section, a comparative study with existing methods in the literature is carried out to further analyze the unique characteristics and advantages of our method.

All simulations were conducted in MATLAB R2024a, we used CasADi~\cite{casadi} to formulate the optimization problems and solve them by IPOPT~\cite{ipopt}. The quadrotor parameters used in the simulations are presented in Table~\ref{tab: quadrotor parameters}. Two scenarios are designed to validate the effectiveness of the proposed method compared to others: 1) a circular reference trajectory with two cylindrical obstacles; 2) a circular reference trajectory with a narrow gap.  

\begin{figure*}
    \centering
    \begin{subfigure}{0.24\textwidth}
        \centering
        \includegraphics[scale=0.24]{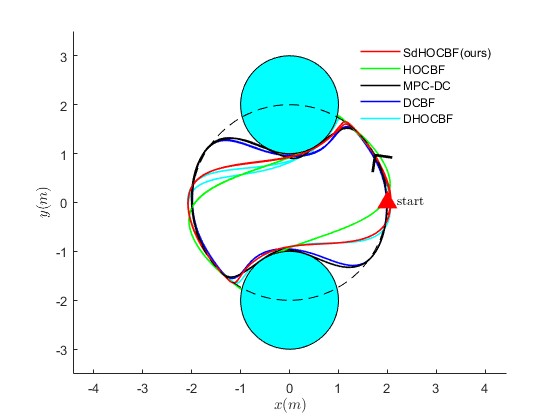}
        \caption{}
        \label{fig: circle trajectories}
    \end{subfigure}
    \begin{subfigure}{0.24\textwidth}
        \centering
        \includegraphics[scale=0.24]{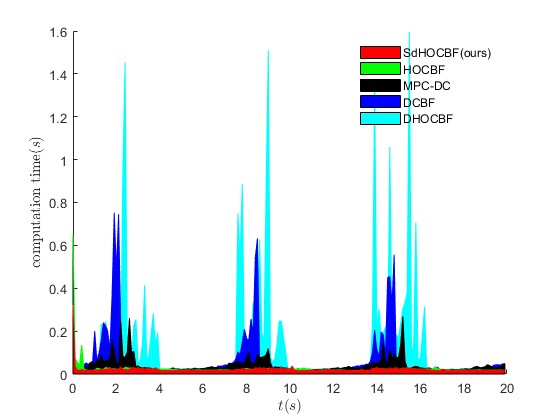}
        \caption{}
        \label{fig: computation time}
    \end{subfigure}
    \begin{subfigure}{0.24\textwidth}
        \centering
        \includegraphics[scale=0.24]{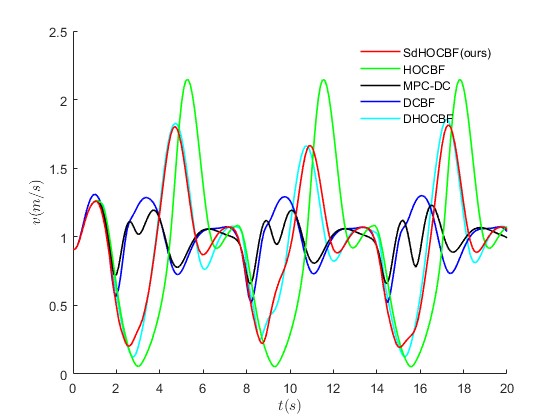}
        \caption{}
        \label{fig: velocity}
    \end{subfigure}
    \begin{subfigure}{0.24\textwidth}
        \centering
        \includegraphics[scale=0.24]{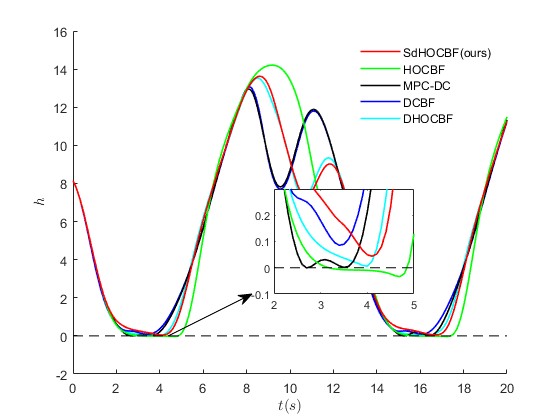}
        \caption{}
        \label{fig: hvalue}
    \end{subfigure}
    \caption{Trajectories, computation time, HOCBF values and velocities from different controllers. (a) Quadrotor's trajectories generated by different controllers. The black dash line, two cyan circles, red triangle and black arrow denote the reference trajectory, cylindrical obstacles, the start point the direction of trajectories, respectively. (b) Computation time to slove MPC optimization of different controllers. (c) the velocities of quadrotor by different controllers. (d) the value of $h(\bm{x})=x^2+(y-2)^2-1$. $h(\bm{x}) \geq 0$ imply the safety of quadrotor.}
    \label{fig:comparison for simulation 1}
\end{figure*}
\begin{figure*}
    \centering
    
    \begin{subfigure}[t]{0.24\linewidth}
        \centering
        \includegraphics[scale=0.23]{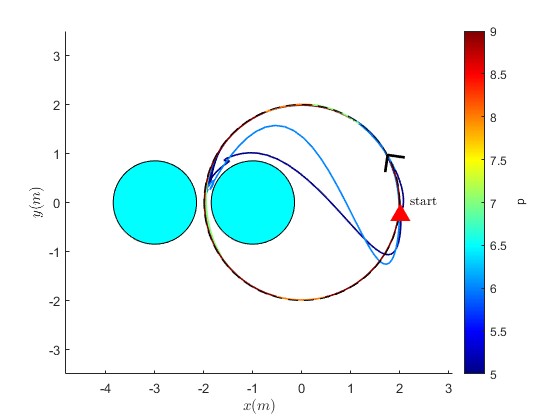}
        \caption{}
        \label{fig: narrow gap sdhocbf}
    \end{subfigure}
    \begin{subfigure}[t]{0.24\linewidth}
        \centering
        \includegraphics[scale=0.23]{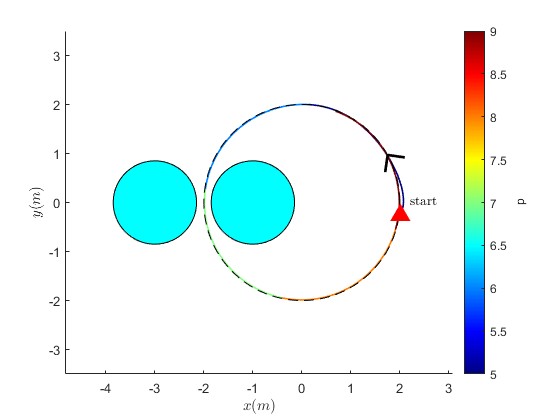}
        \caption{}
        \label{fig: narrow gap hocbf}
    \end{subfigure} 
    \begin{subfigure}[t]{0.24\linewidth}
        \centering
        \includegraphics[scale=0.23]{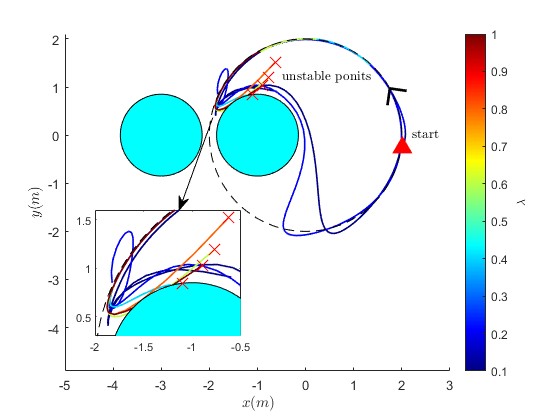}
        \caption{}
        \label{fig: narrow gap dcbf}
    \end{subfigure}   
    \begin{subfigure}[t]{0.24\linewidth}
        \centering
        \includegraphics[scale=0.23]{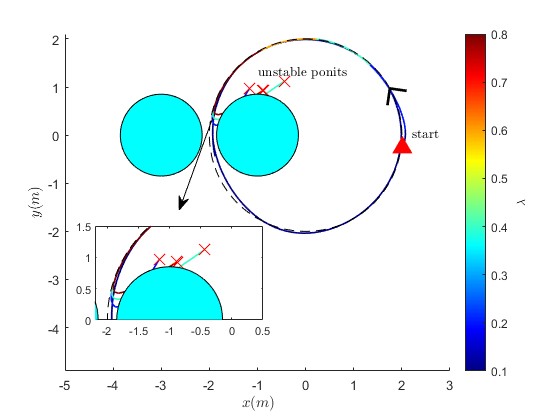}
        \caption{}
        \label{fig: narrow gap dhocbf}
    \end{subfigure}    
    \caption{Trajectories generated by different controllers for the circle with narrow gap. The black arrows indicate the direction along the trajectories. (a) Trajectories generated by SdHOCBF with $p=5,6,7,8,9$. (b) Trajectories generated by HOCBF with $p=5,6,7,8,9$. (c) Trajectories generated by DCBF with $\lambda=0.1,0.2,0.4,0.6,0.8,1$. When $\lambda=1$, DCBF is equivalent to MPC-DC. (d) Trajectories generated by DHOCBF with $\lambda=0.1,0.2,0.4,0.6,0.8$.}
    \label{fig: narrow gap trajectories}
\end{figure*}
For all simulations, the inner-loop controller is executed at a frequency of $1kHz$  and its gains are set to $k_{p,xy}=24,k_{p,z}=0.7,\bm{K}_d = diag(0.8,0.8,0.3)$. Fixed step-time $T=0.1s$ and prediction horizon $N=20$ are used in the outer-loop MPC and the follwing MPC matrix is used for all simulations: $\bm{Q} = diag(100,100,100,1,1,1,1,1,1,1,1,1)$, $\bm{P} =10\bm{Q}$, $\bm{R} = diag(0.01,0.01,0.01)$. We set $\bm{s}_{a,max}=-\bm{s}_{a,min}=[40,40,40]^\top$ to restrict virtual inputs in \eqref{eq: augment position dynamics}.

We compare our proposed method with the following four approaches: a) HOCBF~\cite{xiaohocbf}: remove SdHOCBF constraints in \eqref{eq: mpc controller} and use the quadratic promgram(QP) as a safe filter; b) MPC-DC~\cite{fraschMPCDC}: replace SdHOCBF in \eqref{eq: mpc controller} with distance constraints; c) DCBF~\cite{zengmpccbf}: replace SdHOCBF in \eqref{eq: mpc controller} with discrete control barrier functions (DCBFs) $h(\bm{z}_{i+1|k})-h(\bm{z}_{i|k}) \geq -\lambda h(\bm{z}_{i|k}),\forall i \in \lbrace0,\cdots, N-1\rbrace$; d) DHOCBF~\cite{liumpchocbf}: replace SdHOCBF in \eqref{eq: mpc controller} with discrete control barrier functions (DCBFs) $h(\bm{z}_{i+1|k})-h(\bm{z}_{i|k}) \geq -\lambda h(\bm{z}_{i|k}),i=0,1,3,4$.

\subsubsection{Circle with two cylinders}
To compare the behaviour of the different control methods, we design a cicular reference trajectory defined by flat output $\bar{\bm{y}}(t) = [2\cos(0.5t),2\sin(0.5t),1,0.5t]^T$ and two cylindrical obstacles which can be described by HOCBF as $h_1(\bm{x}) = x^2+(y-2)^2-1$ and $h_2(\bm{x}) = x^2+(y+2)^2-1$, which both have relative degree 4 for system \eqref{eq: augment position dynamics}. For DCBF and DHOCBF, we use $\lambda=0.2$, while for HOCBF and our proposed method, we use all class $\mathcal{K}$ functions $\alpha(x) = px$ with $p=5$ for a fair comparison. All the simulations start at $\bm{x}(0)=[2,-0.25,1,0.4,0.82,0,0,0,0,1,0.5,-0.4,0]^T$ and the simulation duration is set to $T_d=20s$.

The trajectories generated by different methods are shown in Fig.~\ref{fig: circle trajectories}. It shows that DCBF detects obstacles and returns to the reference trajectory faster than all other methods. In contrast, HOCBF and MPC-DC only becomes effective when the quadrotor is very close to the obstacles. SdHOCBF and DHOCBF exhibit intermediate behaviors between these two extremes. However, Fig.~\ref{fig: computation time} shows that the computation time of DCBF and DHOCBF are significantly higher than those of HOCBF and SdHOCBF. Moreover, their computation time are unstable and occasionally fail to meet real-time requirements. This indicates that the performance of DCBF and DHOCBF comes at the cost of computation complexity. In comparison, SdHOCBF achieves similar desired properties without incurring additional computational burden.

Moreover, Fig.~\ref{fig: hvalue} shows that only HOCBF can not guarantee the safety of quadrotor which validates the necessity of introducing SdHOCBFs. In Fig.~\ref{fig: velocity}, the velocity from HOCBF exhibits large fluctuations, which is in contrast to the other three methods. This highlights that incorporating sdHOCBFs as constraints in the MPC framework can achieve a better trade-off between optimality and safety.

\subsubsection{Circle with a narrow gap}
In this simulation, we keep all settings the same as in the previous simulation, except for replacing the two obstacles with $h_1(\bm{x})=(x+3)^2+y^2-0.85^2$ and $h_2(\bm{x})=(x+1)^2+y^2-0.85^2$ in order to create a narrow gap. To evaluate the effectiveness of different methods in this scenario, we conducted simulations for DHOCBF and DCBF with $\lambda = 0.1,0.2,0.4,0.6,0.8$. For DCBF, we also tested $\lambda=1$, which is equivalent to MPC-DC. For HOCBF and SdHOCBF, we use $p=5,6,7,8,9$.

The simulation results are shown in Fig.~\ref{fig: narrow gap trajectories}.  Fig.~\ref{fig: narrow gap dcbf} demonstrates that both MPC-DC and DCBF fail to navigate through the narrow gap. In contrast, Fig.~\ref{fig: narrow gap sdhocbf}~\ref{fig: narrow gap hocbf}~\ref{fig: narrow gap dhocbf} indicate that SdHOCBF, HOCBF, and DHOCBF are all capable of passing through the gap. Among them, HOCBF performs the best, successfully navigating the gap for all tested values of $p$. SdHOCBF ranks second, managing to pass through the gap in most cases, though it may choose to avoid the gap when $p$ is too small. DHOCBF shows the weakest performance, as the quadrotor often becomes unstable in front of the obstacles for most values of $\lambda$.

\section{CONCLUSIONS}
\label{sec: conclusions}
This paper presents a cascaded control framework for safe trajectory tracking of quadrotors. The proposed approach integrates a MPC scheme in the outer loop with SdHOCBFs constraints applied at the first prediction step to generate desired twice differentiable accelerations ensuring safety. The inner-loop controller is designed to prioritize tracking of the safe reference accelerations over yaw orientation. Comprehensive comparisons with existing methods demonstrate the effectiveness of our approach in achieving a better trade-off between safety and optimality without incurring additional computational burden.

\bibliographystyle{IEEEtran}
\bibliography{references}
\end{document}